\documentclass[12pt]{article}
\usepackage[utf8]{inputenc}

\oddsidemargin
-20pt\evensidemargin 0pt

\usepackage{tcolorbox}
\newtcolorbox{mybox}{colback=elgray,boxrule=0mm,top=0.5cm,bottom=0.5cm,left=0.5cm,right=1.5cm,}

\usepackage{mathrsfs}

\usepackage{graphicx}
\usepackage{epstopdf}
\usepackage{amsthm}
\usepackage{amsmath}
\usepackage{amssymb}
\usepackage{amstext}
\usepackage{amsfonts}
\usepackage{enumerate}

\usepackage{float}
\usepackage{pgf}
\usepackage{verbatim}
\usepackage{caption}
\usepackage{graphicx}
\usepackage[english]{babel}

\usepackage[authoryear,nonamebreak,bibstyle]{natbib}

\usepackage{footnote}
\makesavenoteenv{tabular}
\makesavenoteenv{table}

\usepackage[shortlabels]{enumitem}
\usepackage{multirow}
\usepackage[title]{appendix}
\usepackage{xcolor}
\usepackage{setspace}
\usepackage{thmtools,hyperref}

\newtheorem{tm}{Theorem}
\newtheorem{dfn}{Definition}
\newtheorem{lma}{Lemma}
\newtheorem{assu}{Assumptions}
\newtheorem{prop}{Proposition}
\newtheorem{cro}{Corollary}
\newtheorem*{theorem*}{Theorem}
\newcommand{\cor}{\begin{cro}}
\newcommand{\corr}{\end{cro}}
\newtheorem{exa}{Example}
\newcommand{\ex}{\begin{exa}}
\newcommand{\exx}{\end{exa}}
\newtheorem{remak}{Remark}
\newcommand{\rmk}{\begin{remak}}
\newcommand{\rmkk}{\end{remak}}
\newcommand{\thm}{\begin{tm}}

\newcommand{\thmm}{\end{tm}}
\newcommand{\lm}{\begin{lma}}
\newcommand{\lmm}{\end{lma}}
\newcommand{\ass}{\begin{assu}}
\newcommand{\asss}{\end{assu}}
\newcommand{\df}{\begin{dfn}  }
\newcommand{\dff}{\end{dfn}}
\newcommand{\prp}{\begin{prop}}
\newcommand{\prpp}{\end{prop}}
\newcommand{\bqu}{\sloppy \small \begin{quote}}
\newcommand{\equ}{\end{quote} \sloppy \large}

\newcommand{\eq}{\begin{equation}}
\newcommand{\eqq}{\end{equation}}
\newtheorem{claim}{\it Claim}
\newcommand{\cl}{\begin{claim}}
\newcommand{\cll}{\end{claim}}
\newcommand{\bit}{\begin{itemize}}
\newcommand{\eit}{\end{itemize}}
\newcommand{\ben}{\begin{enumerate}}
\newcommand{\een}{\end{enumerate}}
\newcommand{\bcen}{\begin{center}}
\newcommand{\ecen}{\end{center}}
\newcommand{\fn}{\footnote}
\newcommand{\ds}{\begin{description}}
\newcommand{\dss}{\end{description}}
\newcommand{\cs}{\begin{cases}}
\newcommand{\css}{\end{cases}}

\makeatletter
\newcommand{\customlabel}[2]{%
\protected@write \@auxout {}{\string \newlabel {#1}{{#2}{}}}}
\makeatother

\def\reff #1\par{\noindent\hangindent =\parindent
\hangafter =1 #1\par}
\def\title #1{\begin{center}
{\Large {\bf #1}}
\end{center}}
\def\author #1{\begin{center} {\large #1}
\end{center}}
\def\date #1{\centerline {\large #1}}

\def\date #1{\centerline {\large #1}}

\def\R{{\rm I\kern-1.7pt R}}

\let\Large=\large
\let\large=\normalsize

\usepackage{tikz}
\usetikzlibrary{math}

\usepackage{subcaption}

\theoremstyle{plain}
\newtheorem{theorem}{Theorem}[section]

\newtheorem{corollary}[theorem]{Corollary}
\newtheorem*{proposition*}{Proposition}
\newtheorem*{corollary*}{Corollary}
\newtheorem*{subclaim*}{Claim}

\theoremstyle{definition}

\newtheorem*{rem*}{Remark}

\newtheorem*{definition*}{Definition}
\newtheorem*{example*}{Example}
\newtheorem*{fact*}{Fact} 

\theoremstyle{remark}
{\end{minipage}\end{equation}}

{\end{minipage}\end{equation}}

\newenvironment{eqpar*}{\begin{equation*}\begin{minipage}{0.8\columnwidth}}%
{\end{minipage}\end{equation*}}

\DeclareMathOperator{\expectation}{\mathbb{E}}

\newcommand{\bsim}[1]{\mathbin{\sim_{#1}}}
\newcommand{\nbsim}[1]{\mathbin{\not\sim_{#1}}}

\newcommand{\rsim}[1]{\bsim[r]}
\newcommand{\nrsim}[1]{\nbsim[r]}

\providecommand{\int}{\Zset}
\providecommand{\reals}{\Rset}

\DeclareMathOperator{\prob}{\Delta}

\newcommand{\term}[1]{\textbf{#1}}

\newcommand\xqed[1]{%
  \leavevmode\unskip\penalty9999 \hbox{}\nobreak\hfill
  \quad\hbox{{#1}}}

\newcommand\qedthm{\xqed{${\scriptscriptstyle\blacksquare}$}}

\usepackage{mathtools}
\usepackage{array}
\usepackage{nicefrac}
\usepackage{bbold}
\usepackage{nicematrix}
\usepackage{bigints}
\usepackage{thmtools}

\DeclareSymbolFont{AMSb}{U}{msb}{m}{n}
 \DeclareMathSymbol{\Nset}{\mathbin}{AMSb}{"4E}
 \DeclareMathSymbol{\Zset}{\mathbin}{AMSb}{"5A}
 \DeclareMathSymbol{\Rset}{\mathbin}{AMSb}{"52}
 \DeclareMathSymbol{\Qset}{\mathbin}{AMSb}{"51}
  \DeclareMathSymbol{\Fset}{\mathbin}{AMSb}{"46}
 \DeclareMathSymbol{\Cset}{\mathbin}{AMSb}{"43}
 \DeclareMathSymbol{\Kset}{\mathbin}{AMSb}{"4B}
 
 \DeclareMathSymbol{\Sset}{\mathbin}{AMSb}{"53}

 \usepackage{tikz-3dplot} 
 \usetikzlibrary{calc}

 \definecolor{Xred}{RGB}{204,102,102}
\definecolor{Xgreen}{RGB}{102,153,102}
\definecolor{Xblue}{RGB}{51,102,204}
\definecolor{Xgold}{RGB}{204,153,51}

\definecolor{AcceptCone}{HTML}{0E8F7D}
\definecolor{RejectCone}{HTML}{E46C16}
\definecolor{BallpenBlue}{HTML}{1B4997}
\definecolor{UniBlue}{HTML}{113777}
\definecolor{DkGreen}{HTML}{0E8F7D}
\definecolor{dcyan}{RGB}{0,55,55}
\definecolor{elgray}{RGB}{245,245,245}

\usepackage{pgfplots}
\usetikzlibrary{fadings}
\pgfplotsset{compat=newest}
\usetikzlibrary{intersections,shapes.geometric} 
\usepgflibrary{shapes.arrows} 
\usetikzlibrary{shapes,arrows,backgrounds,positioning,calc,fit}
\usetikzlibrary{patterns}
\tikzstyle{dot}=[fill=black,circle,minimum size=3pt,inner sep=0pt]
\tikzstyle{label}=[draw=none,fill=none, minimum size=0, inner sep=0]
\tikzstyle{evidence}=[fill=black!25]
\tikzstyle{query}=[fill=black,text=white]
\usepgfplotslibrary{ternary, units}

\definecolor{pnas}{RGB}{185,69,35}
\definecolor{pnasblue}{RGB}{12,100,180}
\definecolor{el}{RGB}{147,72,29}
\definecolor{elo}{RGB}{191,90,35}

\hypersetup{
    colorlinks=true,
    linkcolor=el,
    filecolor=el,      
    urlcolor=el,
    citecolor=el,
    }

\usepackage[hang,marginal]{footmisc}

\usepackage{framed}
\theoremstyle{plain}
\usepackage{orcidlinkNEW}
\newtheorem{prototheorem}{Theorem}[section]
\newtcolorbox{mypbox}{colback=orange!8,boxrule=0mm,top=10pt,bottom=10pt,left=5pt,right=1pt,}
\newenvironment{cthm}
   {\begin{mypbox}\begin{prototheorem}}
   {\end{prototheorem}\end{mypbox}}

\begin{document}

\begin{titlepage}
\def\thefootnote{\fnsymbol{footnote}}
\vspace*{0.8in}

\title{\scshape\fontsize{12}{14}\selectfont {\bf \scshape\fontsize{12}{14}\selectfont  Two-Person Adversarial Games are Zero-Sum:\\  \vspace{0.5em}
 An Elaboration of a Folk Theorem}\fn{\noindent The authors are grateful to Samuel A. Alexander, Pradeep Dubey, Marcelo Ariel Fernandez, Ani Ghosh, Michael Grossberg, David Kellen, Yehuda John Levy, Conor Mayo-Wilson, Andrew McKenzie, Lawrence S. Moss, Aristotelis Panagiotopolous, Rohit Parikh, Marcus Pivato, Andrew Powell, Ram Ramanujam, Fedor Sandomirskiy, Amnon Schreiber, Wolfgang Spohn, Jack Stecher, Maxwell Stinchcombe, and Metin Uyanik for conversation and continuing collaboration on the subject of this letter.  Khan's work
was supported by the 2022 Provost's JHU Discovery Award ``Deception and Bad-Faith Communication.'' Aspects of this work were presented for the Research Seminar in Set Theory at the University of Vienna (March 14, 2024), the CUNY Graduate Center Seminar in Philosophy, Logic and Games (\textsc{philog}; March 29, 2024), and \textit{The 35th International Conference on Game Theory} at Stony Brook (July 16, 2024).
The authors extend a final thanks to an anonymous reviewer for a careful reading of the manuscript.}}
 \vskip 1.5em
 
 \author{\orcidlinki{M. Ali Khan}{0000-0003-3852-4400}\,\footnote{Department of Economics, The Johns Hopkins University\\  {\bf email} {akhan@jhu.edu }}, \orcidlinki{Arthur Paul Pedersen}{0000-0002-2164-6404}\,\footnote{Department of Computer Science, Remote Sensing Earth Systems Institute, The City University of New York\\ {\bf email}
 {app@arthurpaulpedersen.org}} \footnote{Corresponding author} 
  and  \orcidlinki{David Schrittesser}{{0000-0002-4622-2675}}\,\footnote{Institute for Advanced Study in Mathematics, Harbin Institute of Technology\\  {\bf email}     {david.schrittesser@univie.ac.at} }   }

\vskip 1.00em
\date{July 31, 2024}

\vskip 1.75em

\vskip 1.00em

\baselineskip=.18in

\noindent{\bf Abstract.} The observation that every two-person adversarial game is an affine transformation of a zero-sum game is traceable to \citet{lr57} and made explicit in \citet{au87Palgrave}. Recent work of \citet{adp09} and of \citet{ra23} in increasing generality, proves what has so far remained a conjecture. We present two proofs of an even more general formulation: the first draws on multilinear utility theory developed by \citet{fr78}; the second is a consequence of \citeauthor{adp09}'s \citeyearpar{adp09} proof itself for a special case of a two-player game in which each player has a set of three actions.

\vskip 2em
{\fontsize{10}{11}\selectfont
\noindent {\it Journal of Economic Literature} Classification
Numbers: C72, D01. 

\vskip 0.8em

\noindent {\it 2020 Mathematics Subject} Classification Numbers: 91A05,  91A10,  91A30.

\vskip .8em


\vskip .8em

\noindent {\it Running Title}:\quad
Two-Person Adversarial Games are Zero-Sum 
}

\end{titlepage}








 


%
%
%
%

\setcounter{footnote}{0}
%
%
%


%

\pagebreak

\vspace*{-50pt}

\begin{mybox}{\bfseries\fontsize{14}{16}\selectfont Highlights}{\fontsize{12}{14}\selectfont
\begin{itemize}[itemsep=1em,rightmargin=1pt,labelsep=0pt,leftmargin=20pt,labelwidth=18pt,topsep=1.5em,itemindent=0pt,align=left]
    \item[{\fontsize{9}{10}\textbullet}] Theorem casting Luce-Raiffa--Aumann conjecture in most general form to date

    \item[{\fontsize{9}{10}\textbullet}]  Theorem entailing as corollaries results of \citet{adp09} and \citet{ra23}

\item[{\fontsize{9}{10}\textbullet}]  Novel proof adapting multilinear utility theory of \citet{fr78}
, throwing light on relationship that has escaped treatment in prior literature

\item[{\fontsize{9}{10}\textbullet}]  Alternative proof appealing to result of \citet{adp09} for two-player game in which each player is limited to three actions, an approach missed in earlier work

\item[{\fontsize{9}{10}\textbullet}]  Formal \& thematic ties to  \citet{mv78} 
on strategically zero-sum games

\item[{\fontsize{9}{10}\textbullet}]  
Implications for complexity of computing Nash equilibria \'{a} la \citet{dgp09}

\end{itemize}
}
\end{mybox}

\vspace*{10pt}

\bqu In all of man's written record there has been a preoccupation with conflict
of interest; possibly only the topics of God, love, and inner struggle have received comparable attention. The scientific study~\ldots~comprises a small, but growing, portion of this literature. As a reflection of this trend, we find today that conflict of interest, both among individuals
and among institutions, is one of the more dominant concerns of at least several of our academic departments: economics, sociology, political
science, and other areas to a lesser degree.

\hspace{ \stretch{3} } ---\citet{lr57}\equ 

\vspace*{10pt}

\bqu   Game theory,  a mathematical theory that  shares a common foundation in
the assumption  that actors must be strategic, or individualistically
competitive against others, offers a unified methodology and a comprehensive
understanding of purposive agency that rejects joint maximization and shared
intention, and reduces preference satisfaction to narrow self-interest.\footnote{The author notes, ``Theorists apply the same tools and models to widely divergent fields of investigation:  economics, politics, conflict resolution and evolution.'' The  epigraph is a composite taken from  sentences in~\citet[pp. vxiii-xx]{am15}. The footnotes of the author's prologue give an up-to-date bibliography on the subject enjoying extensive interdisciplinary reach.}

\hspace{ \stretch{3} } ---\citet{am15}  \equ

\section{Introduction}

A folk theorem of central importance to the theory of games states that every two-person adversarial game is an affine transformation of a zero-sum game.  
\citet{lr57} appeal to this observation in their exposition of two-person non-cooperative games to justify confining attention to zero-sum games in their treatment of adversarial games --- called by them \textit{strictly competitive} games.

Decades later, \citet{au87Palgrave} gave the contours of a formal description of the folk theorem in his masterful 1987 survey of game theory:
\begin{quote}
\small
Recall that a strictly competitive game is defined as a two-person game in which if one outcome is preferred to another by one player, the preference is reversed for the other.
Since randomized strategies are admitted, this condition applies also to mixed outcomes (probability mixtures of pure outcomes). From this it may be seen that a two-person game is strictly competitive if and only if, for an appropriate choice of utility functions, the utility payoffs of the players sum to zero in each square of the matrix (p. 13).\footnotemark
\end{quote}\footnotetext{\label{fn:2}We remind the reader that, in general, the mixed extension of an adversarial game is \textit{not} adversarial.  This letter's principal result, \autoref{t.fishburn}, accordingly casts the folk theorem in a form that is correct.}
More recently,~\citet{adp09}, henceforth  ADP, have complained that  the
literature treating the folk theorem has sown confusion and further that  no
proof exists in full or in outline.
They proceed to give two elementary proofs of the folk theorem for finite action sets, one in an algebraic register, the other in a combinatorial one.
 
Subsequent work by~\citet{ra23} takes up the folk theorem, calling it the
Luce-Raiffa--Aumann (LRA) conjecture, as we shall  henceforth call the folk
theorem.\fn{For precise renderings of the LRA conjecture, see \autoref{t.fishburn} and earlier expressions
in~\citet{au87Palgrave} (but see \autoref{fn:2}), \citet{adp09},~and~\citet{ra23}.} \cite{ra23}  proves the LRA conjecture for the case of closed interval action sets and continuous utilities, 
extending the result of ADP for finite action sets. He offers a functional analytic proof by necessity different from the previous ones.

 This letter offers a synthetic treatment, and thereby  advances the following contributions to the theory of  two-person games:

\begin{itemize}[itemsep=1em,leftmargin=3em,labelsep=.2in,topsep=1em,rightmargin=30pt]
\item[{\fontsize{9}{10}$\blacktriangleright$}]  A theorem (\autoref{t.fishburn}) that casts the LRA conjecture in its
most general form to date and has as corollaries the central results of \citet{adp09} and \citet{ra23}
\item[{\fontsize{9}{10}$\blacktriangleright$}] A novel proof of the aforementioned theorem that adapts results in multilinear utility theory due to \citet{fr78}, a relationship that has escaped treatment in prior literature
\item[{\fontsize{9}{10}$\blacktriangleright$}] An alternative proof (\autoref{ss.altproof}) that appeals to the result of \citet{adp09} for a two-player
game in which each player is limited to three actions, an approach missed in earlier work
\item[{\fontsize{9}{10}$\blacktriangleright$}] Formal and thematic connections to developments of \citet{mv78} on strategically zero-sum games
\item[{\fontsize{9}{10}$\blacktriangleright$}] Implications for the study of complexity of computing Nash equilibria \'{a} la \citet{dgp09}
\end{itemize}

 \def\sectionautorefname{Section}
  \def\subsectionautorefname{Section}

The remainder of this letter is set forth as follows: 
\autoref{sec:2} lays out the basic notation and terminology necessary for
presenting the theorem due to \citet{fr78} that is used in this paper.  \autoref{s.fishburn} formulates and proves
our main result (\autoref{t.fishburn}) and uses it to derive the recent theorems of  ADP (\autoref{corollary:adp09})  and Raimondo (\autoref{cor:ra23}); this section also gives an alternative proof of our main result (\autoref{ss.altproof}). \autoref{TeXFolio:sec4} concludes the letter by
drawing connections to~\citet{mv78} on strategically zero-sum games and recent work in both computer science and economics on
computing equilibria of two-player games.

\section{Notation and Terminology}
\label{sec:2}

Given a binary relation $\precsim$ on a set $X$, we denote its 
asymmetric and symmetric parts on $X$ by $\prec$ and $\sim$, respectively.  Given an $n$-tuple $(p_{1},\ldots,p_{n})$ belonging to the $n$-fold Cartesian product of $n$ sets and $i=1,\ldots, n$, we adopt the usual convention to denote by $p_{-i}$ the $n-1$-tuple  $(p_1, \ldots, p_{i-1},p_{i+1},\ldots,p_{n})$ and by $(q, p_{-i})$ the $n$-tuple $(p_1, \ldots, p_{i-1}, q, p_{i+1},\ldots p_{n})$.

An \term{ordered bilinear mixture space} is a quintuple of the form $\mathcal{M}=\bigl\langle M_{1},M_{2},\,\precsim,\, \oplus_{1},\oplus_{2}\bigr\rangle$ consisting of nonempty sets $M_{1}$ and $M_{2}$, a binary relation $\precsim$ on $M_{1}\times M_{2}$, and functions $\oplus_{i}$  for $i=1,2:$
\begin{align*}
\oplus_{i}: [0,1]\times M_{i}\times M_{i}\,\to\, M_{i}&\qquad 
(\alpha,p,q)\,\mapsto\, \alpha p\,\oplus_{i}\,(1-\alpha) q 
\end{align*}
satisfying the following five requirements for all $p,q,r,s\in M_{1}\times M_{2}$, $\alpha,\beta\in [0,1]$, and $i,j=1,2$:

\begin{itemize}[itemsep=1.5em,leftmargin=5em,labelsep=.35in,topsep=1em]
\item[{\textsc{ms}1}] {$\precsim$} is a total preorder on $M_{1}\times M_{2}$;

\item[{\textsc{ms}2}] $\Bigl( \,\alpha p_{i}\oplus_{i}(1-\alpha) q_{i}, ~r_{-i}\,\Bigr) \;\sim\;\Bigl(\, (1-\alpha) q_{i}\oplus_{i}\alpha p_{i}, ~r_{-i}\,\Bigr)$

\item[{\textsc{ms}3}] $\biggl(\,\beta\bigl(\alpha p_{i}\oplus_{i}(1-\alpha) q_{i}\bigr)\oplus_{i}\bigl(1-\beta\bigr)q_{i}, ~r_{-i}\,\biggr) \;\sim\;\Bigl(\,\alpha\beta p_{i}\oplus_{i} (1-\alpha\beta)q_{i}, ~r_{-i}\,\Bigr)$

\item[{\textsc{ms}4}] If $p\,\prec\, q$ and $q \,\prec\, (r_{i}, p_{-i})$, then there are $\alpha,\beta\in(0,1)$ such that:

\vspace*{2ex}

\hspace*{4em}$\Bigl( \,\alpha p_{i}\oplus_{i}(1-\alpha) r_{i}, ~p_{-i}\,\Bigr)\;\prec\; q$ \qquad and \qquad $q\;\prec\;\Bigl( \,\beta p_{i}\oplus_{i}(1-\beta) r_{i}, ~p_{-i}\,\Bigr)$ 

\item[\textsc{ms}5]  If $p\,\prec\, q$ and $(r_{i}, p_{-i})\;\sim\; (s_{j},q_{-j})$, then:

\vspace*{2ex}

\hspace*{14em}$\Bigl(\,\alpha p_{i}\oplus_{i} (1-\alpha)r_{i}, ~ p_{-i}\,\Bigr)\;\prec\; \Bigl(\,\alpha q_{j}\oplus_{j} (1-\alpha) s_{j}, ~q_{-j}\,\Bigr)$.

\end{itemize}
Conditions \textsc{ms}2, \textsc{ms}3, \textsc{ms}4 are the usual mixture space axioms along each dimension, while condition \textsc{ms}5 is an interdimensional reformulation of the usual independence axiom.

A function
$u:M_{1}\times M_{2}\to\Rset$ is said to be a \term{bilinear representation} of $\precsim$ if it satisfies the following two properties for all $p,q,r\in M_{1}\times M_{2}$, $\alpha\in[0,1]$, and $i=1,2$:
\begin{itemize}[itemsep=1em,leftmargin=6em,labelsep=.4in,topsep=1em]
    \item[\textsc{Rep}] $p\,\precsim\, q$\qquad if and only if\qquad $u(p)\,\leq\, u(q)$; \quad and
    \item[\textsc{Bilin}]  $u\bigl(\,\alpha p_{i}\oplus_{i}(1-\alpha)q_{i}, ~r_{-i}\,\bigr)\quad=\quad \alpha u(p_{i}, ~r_{-i})\;+\;(1-\alpha)u(q_{i}, !r_{-i})$.
    \end{itemize}
Let $X$ be a set. Recall that a function $f:X\to\Rset$ is said to be a \term{positive affine transformation} of a function $g:X\to\Rset$ if there are $\alpha,\beta\in \Rset$ with $\alpha>0$ such that $f(x)=\alpha g(x)\,+\,\beta$  for all $x\in X$.

We can now present: 

\begin{theorem}[\citealt{fr78}]\label{thm:fr78} Suppose $\mathcal{M}=\bigl\langle M_{1},M_{2},\,\precsim,\, \oplus_{1},\oplus_{2}\bigr\rangle$ is an ordered bilinear mixture space.  Then there is a bilinear representation $u$ of $\precsim$ such that any other bilinear representation of $\precsim$ is a positive affine transformation of $u$ --- that is to say, the function $u$ is unique up to a positive affine transformation.

\end{theorem}

\begin{proof}  See proof of Theorem 2 by \citet{fr78}.   \end{proof}

\section{The Luce-Raiffa-Aumann  Conjecture Reformulated }\label{s.fishburn}
A \term{two-person non-cooperative game} is a quadruple of the form $\mathcal{G}=\bigl\langle\mathcal{P}_{1},\mathcal{P}_{2}, u_{1},u_{2}\bigr\rangle $ consisting of nonempty sets $\mathcal P_1$ and $\mathcal P_2$ and real-valued functions $u_{1}$ and $u_{2}$ on the Cartesian product $\mathcal P_1 \times \mathcal P_2$.

A game $\mathcal{G}$ is said to be \term{adversarial}, or \term{strictly competitive}, if  for all $\sigma,\tau\in  \mathcal{P}_1 \times \mathcal{P}_2$:
\begin{equation}\label{e.str.comp}
u_{1}(\sigma) \;\geq \;u_{1}(\tau)\quad \iff\quad u_{2}(\sigma) \;\leq\; u_2(\tau).
\end{equation}
A game $\mathcal{G}$ is said to be \term{zero-sum} if equality $u_{1}(\sigma)\,+\,u_{2}(\sigma)=0$ obtains for all $\sigma\in \mathcal{P}_{1}\times \mathcal{P}_{2}$. 

If for a given game $\mathcal{G}=\bigl\langle\mathcal{P}_{1},\mathcal{P}_{2}, u_{1},u_{2}\bigr\rangle $ each $\mathcal{P}_{i}$ is convex, the game is said to be \term{bilinear}, or \term{bi-affine}, if each function $u_{i}$ is linear in each coordinate --- that is, for all $\sigma,\tau\in \mathcal{P}_{1}\times \mathcal{P}_{2}$, $\alpha,\beta\in[0,1]$, and $u\in\{u_{1},u_{2}\}$:
\begin{align*}
u\Bigl(\,\alpha \sigma_{1}+(1-\alpha)\tau_{1},\beta \sigma_{2} +(1-\beta)\tau_{2}\,\Bigr)&=\alpha u\Bigl(\, \sigma_{1},~\beta \sigma_{2}+(1-\beta)\tau_{2}\,\Bigl)\,\;+\;\,(1-\alpha) u\Bigl(\, \tau_{1},~\beta \sigma_{2}+(1-\beta)\tau_{2}\,\Bigr)\\
&=\beta u\Bigl(\,\alpha \sigma_{1}+(1-\alpha)\tau_{1},~\sigma_{2}\,\Bigr)\;+\,(1-\beta) u\Bigr(\,\alpha \sigma_{1}+(1-\alpha)\tau_{1},~\tau_{2}\,\Bigr).
\end{align*}

\subsection{Two-Person Adversarial Games are Zero-Sum}
 Established forthwith is that \emph{every two-person adversarial game is, up to a positive affine transformation, zero-sum}, by appeal to \autoref{thm:fr78}.

\begin{cthm}\label{t.fishburn}
Let $\mathcal{G}=\bigl\langle\mathcal{P}_{1},\mathcal{P}_{2}, u_{1},u_{2}\bigr\rangle$  be a two-person non-cooperative game.  Suppose $\mathcal P_1$ and $\mathcal P_2$ are convex and $\mathcal{G}$ is bilinear. Then the following are equivalent:
\begin{itemize}[itemsep=0.5em,leftmargin=2em,labelsep=.1in,topsep=1em]
\item[\textup{(a)}] Game $\mathcal{G}$ is adversarial; 
\item[\textup{(b)}] Function $u_{2}$ is a positive affine transformation of function $-u_{1}$; 
\item[\textup{(c)}] There is a positive affine transformation $v_{1}$ of $u_{1}$ such that $\mathcal{Z}=\langle \mathcal{P}_{1},\mathcal{P}_{2},v_{1},u_{2}\rangle$ is zero-sum.\qedthm
\end{itemize}
\end{cthm}

\begin{proof} To show that (a) implies (b), suppose $\mathcal{G}$ is adversarial.  Define a binary relation $\precsim$ on $\mathcal{P}_{1}\times \mathcal{P}_{2}$ by setting for all $\sigma,\tau\in \mathcal{P}_{1}\times \mathcal{P}_{2}$:
\begin{equation*}
\sigma\;\precsim\; \tau\quad\iff\quad -u_{1}(\sigma) \;\leq -u_{1}(\tau).
\end{equation*}
 By stipulation, the function $-u_{1}$ is a bilinear representation of $\precsim$.  Since $\mathcal{G}$ is adversarial, the function $u_{2}$ is also a bilinear representation of $\precsim$. Thus, by \autoref{thm:fr78}, it follows that  $u_{2}=  \alpha(-u_{1}) \,+\,\beta$ for some $\alpha,\beta\in\Rset$ with $\alpha > 0$, whence (b).

 To see that (b) implies (c), suppose $u_{2}$ is a positive affine transformation of $-u_{1}$, whereby $u_{2}= -\alpha u_{1} \,+\,\beta$ for some $\alpha,\beta\in\Rset$ with $\alpha > 0$.  Observe that $\mathcal{Z}=\langle \mathcal{P}_{1},\mathcal{P}_{2},\alpha u_{1} -\beta,u_{2}\rangle$ is zero-sum, as desired.  

For the implication from (c) to (a), it is straightforwardly verified that if $\mathcal{Z}=\langle \mathcal{P}_{1},\mathcal{P}_{2},v_{1},u_{2}\rangle$ is zero-sum for some positive affine transformation $v_{1}$ of $u_{1}$, then game $\mathcal{G}$ is a adversarial.
\end{proof}

\subsection{Antecedent Results as Corollaries}
Given a finite set of actions $S$, let  $\Delta(S)$ denote the set of all simple probability mass function on $S$.  Given a function $\nu:S_{1}\times S_{2}\to \Rset$, let $\mathbb{E}_{\nu}$ denote expected utility  $\mathbb{E}_{\nu} :\Delta(S_{1})\times \Delta(S_{2})\to\Rset$ given by requiring for all $p\in \Delta(S_{1})\times \Delta(S_{2})$:
\begin{align*}
\mathbb{E}_{\nu}(p)\quad&\coloneqq\quad \sum_{(s_{1},s_{2})\in S_{1}\times S_{2}}\!p_{1}(s_{1})p_{2}(s_{2})\nu(s_{1},s_{2}).
\end{align*} 
Now consider $\mathcal G =  \bigl\langle \Delta(S_{1}),\Delta(S_{2}),\mathbb{E}_{\nu_{1}},\mathbb{E}_{\nu_{2}}\bigr\rangle$ for 
real functions $\nu_{1},\nu_{2}$ on $S_{1}\times S_{2}$. Observe that $\mathcal{G}$ is a bilinear two-person non-cooperative game.
An immediate corollary of \autoref{t.fishburn} is the main result reported in \citep[Theorem 1]{adp09}, which we state without proof. 
\begin{corollary}[\citealt{adp09}]\label{corollary:adp09} Let
$\mathcal G = \bigl\langle \Delta(S_{1}),\Delta(S_{2}),\mathbb{E}_{\nu_{1}},\mathbb{E}_{\nu_{2}}\bigr\rangle$
be an
adversarial game based on finite action sets $S_1$ and $S_2$, as above. 
Then $\nu_{2}$ is an positive affine transformation of $-\nu_{1}$.

\end{corollary}

We may adopt similar notation to formulate the central result reported in \citep{ra23}. 
Let $\Delta_{\lambda}\bigl([0,1]\bigr)$ denote the set of probability measures on $[0,1]$ which are absolutely continuous with respect to Lebesgue measure. 
Given a continuous function $\xi:[0,1]\times[0,1]\to\Rset$, use  $\expectation_\xi$ to denote expected utility, so that  $\expectation_\xi : \;\mathrm{\Delta}_{\lambda}\bigl([0,1]\bigr)\times \mathrm{\Delta}_{\lambda}\bigl([0,1]\bigr)\to\Rset$
and for all $p \in \mathrm{\Delta}_{\lambda}\bigl([0,1]\bigr)\times \mathrm{\Delta}_{\lambda}\bigl([0,1]\bigr)$:
$$
\expectation_\xi(p)\quad \coloneqq\quad \bigintsss\xi(x_{1},x_{2})\,(p_{1}\otimes p_{2})\,(dx_{1}dx_{2}).
$$
As before, let $\mathcal G \coloneqq \bigl\langle \Delta_{\lambda}\bigl([0,1]\bigr),\Delta_{\lambda}\bigl([0,1]\bigr),\mathbb{E}_{\xi_{1}},\mathbb{E}_{\xi_{2}}\bigr\rangle$, and observe that $\mathcal G$ is a bilinear two-person non-cooperative game.   A  corollary of  \autoref{t.fishburn}, stated next, is the main result reported in \citet{ra23}.

\begin{corollary}[\citealt{ra23}]\label{cor:ra23} Suppose $\mathcal G = \bigl\langle \Delta_{\lambda}\bigl([0,1]\bigr),\Delta_{\lambda}\bigl([0,1]\bigr),\mathbb{E}_{\xi_{1}},\mathbb{E}_{\xi_{2}}\bigr\rangle$ is an adversarial game
based on a common action set $[0,1]$
for continuous real functions $\xi_1$ and $\xi_2$, as above.
Then $\xi_{2}$ is an positive affine transformation of $-\xi_{1}$. 
\end{corollary}
\begin{proof}
By \autoref{t.fishburn} we obtain immediately that
there are there are $\alpha,\beta\in\Rset$ with $\alpha >0$ such that such that for every
$p\in \Delta_{\lambda}\bigl([0,1]\bigr)\times \Delta_{\lambda}\bigl([0,1]\bigr)$,
$u\xi_1(p) = -\alpha u\xi_2(p) + \beta$.
By continuity of $\xi_1$ and $\xi_2$,
in fact $\xi_1(p) = -\alpha \xi_2(p) +\beta$.
\end{proof}

\subsection{An Alternative Proof}
\label{ss.altproof}
A natural question is whether   \autoref{t.fishburn} can be obtained from \autoref{corollary:adp09}, especially since the proof of the latter given by~\citet{adp09}
is (at least on the surface) quite different from the one given here, as
well as from the proof by~\citet{fr78}.
We turn to this question next, giving an alternative proof of our main result,
\autoref{t.fishburn}, presupposing its consequence for finite actions sets, that is,~\autoref{corollary:adp09}.

To this end, the following terminology will be useful for the proof.
Given $\mathcal{P} \subseteq \mathcal{P}_1\times \mathcal{P}_2$, 
call pair $(\alpha, \beta) \in \reals\times\reals$ \term{compatible with} $\mathcal{P}$
if$ 
u_2 (p) \;=\; -\alpha  u_1 (p) \,+\, \beta$ for all $p \in \mathcal{P}$.

\begin{proof}
We will concentrate on the non-trivial implication, 
(a) $\Rightarrow$ (b).
To this end, we make the following observation by appealing to ~\autoref{corollary:adp09}.
\begin{subclaim*}
For any three strategy profiles $\bigl\{\,p^{(1)},~p^{(2)}, ~p^{(3)}\,\bigr\} \subseteq \mathcal P_1 \times \mathcal P_2$,
there is a pair $(\alpha,\beta) \in \reals\times\reals$ with $\alpha > 0$  
that is compatible with $\bigl\{\,p^{(1)},~p^{(2)}, ~p^{(3)}\,\bigr\}.$
\end{subclaim*}

To establish this claim,
consider the two-player game with finite action sets
$S_i = \bigl\{\,p^{(1)}_i, ~p^{(2)}_i, ~p^{(3)}_i\,\bigr\}$ for $i=1,2$, and define utilities $u'_1, u'_2$ on $\prob(S_1) \times \prob(S_2)$ to be the restrictions of  
$u_1, u_2$  to $\prob(S_1) \times \prob(S_2)$.
We thereby obtain an adversarial game
$\langle \prob(S_1), \prob(S_2), u'_1, u'_2\rangle$.
By \autoref{corollary:adp09}, 
there is a pair $(\alpha, \beta) \in \reals\times\reals$ 
with $\alpha > 0$ such that~ $u'_2(p) \,=\, -\alpha u'_1(p) \,+\, \beta$
for all $p \in \prob(S_1)\times \prob(S_2)$.
In particular, the pair $(\alpha, \beta)$ is compatible with $\{\,p^{(1)}, ~p^{(2)}, ~p^{(3)}\,\}$.

\autoref{t.fishburn} now is immediate:
Consider a two-person adversarial game
$\mathcal{G}=\langle \mathcal{P}_{1},\mathcal{P}_{2},u_{1},u_{2}\rangle$.
We may assume $u_1$ is not constant.
Fix $p^{(1)}, p^{(2)}$ with
$u_1(p^{(1)}) \neq u_1(p^{(2)})$.
There exists exactly one $(\alpha, \beta) \in \reals\times\reals$ that is compatible with $\{\,p^{(1)}, ~p^{(2)}\,\}$; denote this unique pair by $(\alpha^{*}, \beta^{*})$.
Now let $p^{(3)} \in \mathcal P_1 \times \mathcal P_2$ be arbitrary. 
By the claim, there is a pair $(\alpha,\beta)\in\reals\times\reals$ that is compatible
with $\{\,p^{(1)}, ~p^{(2)}, ~p^{(3)}\,\}$.
We have $(\alpha, \beta) = (\alpha^{*}, \beta^{*})$ by the uniqueness property of the latter.
As  $u_2 (p^{(3)})  = -\alpha^{*} u_1 (p^{(3)}) + \beta^{*}$ and $p^{(3)}$ was arbitrary, we are done.
\end{proof}

\section{Concluding Remarks}
\label{TeXFolio:sec4}

Two directions proceed from the connections between the game theory and the decision theory communities  which have been forged  in this letter.  We leave both for future work. 

The first is a consequence of the observation that  the central results
presented here for convex spaces   may be suitably recast in terms of mixtures
spaces as originally pioneered by~\citet{hm53}. This being said, the question
arises as to whether  the theory of normal form games articulated
by~\citet{na50,na51} and~\citet{de52}  can be set in mixture spaces.   This
would require  suitable embedding theorems that take action sets in a
mixture-space setting to topological vector spaces, and then bringing  back the
existence results available there.\footnote{In this connection, a direct proof
of the results due to \citet{fr78} may be useful for workers in the field. We
are grateful to an anonymous reviewer for suggesting this.} 

Second, in an influential paper,~\citet{mv78} introduce the notion of 
\emph{strategically equivalent games} in the context of two-player games, and
use it to characterize games that are  equivalent in this sense to a zero-sum
game.  The connection to these ideas also merits further investigation and
resolution, especially so since the Moulin--Vial notion  plays 
 a crucial role in the burst of recent activity investigating computational
aspects of two-player games by both economists and computer scientists. The
second proof of the main result of this letter has direct relevance in this
context in view of work stemming from \citet{adp09}.\footnote{For the economic literature,
see, e.g., ~\citet{set02,ss06}~and~\citet{te23}, and their references to earlier
work:~\citet{lh64,vo58} and~\citet{im80}. For the literature in computer
science,
see, e.g., ~\citet{ks12,ph17,hy19,hg23,dgp09,cdt09}
and~\citet{nrtv}. } 

A final  summary statement. In its focus on two-player games, the basic thrust
of this work  goes against the grain of the development of non-cooperative game
theory in which the generalization of two-player games was sought in
$n$-player games, and even in games with an uncountable continuum
of players.\footnote{These are games in which each agent from a continuum is
strategically-negligible but a statistical summary of the plays of all the
players has an impact on an individual decision; see~\citet{ks02} and their
references.}  The adversarial aspect that is emphasized here concentrates on
the them-versus-us aspect without any defensiveness.

\bibliographystyle{plainnat}

\end{document}